\documentclass[conference]{IEEEtran}
\usepackage[cmex10]{amsmath}
\usepackage{amssymb}
\usepackage{bbm,dsfont}
\newtheorem{theorem}{Theorem}

\newtheorem{definition}[theorem]{Definition}
\newtheorem{example}{Example}
\newtheorem{remark}{Remark}
\begin{document}
\sloppy
\title{Multiset Codes for Permutation Channels}%
\author{\IEEEauthorblockN{Mladen Kova\v cevi\' c and Dejan Vukobratovi\'c}
        \IEEEauthorblockA{Department of Power, Electronics, and Communication Engineering, \\
                          University of Novi Sad, 21000 Novi Sad, Serbia \\
                          Emails: \{kmladen, dejanv\}@uns.ac.rs }}%
\maketitle
\begin{abstract}
This paper introduces the notion of \emph{multiset codes} as relevant to 
the problem of reliable information transmission over permutation channels. 
The motivation for studying permutation channels comes from the effect of out 
of order delivery of packets in some types of packet networks. The proposed 
codes are a generalization of the so-called \emph{subset codes}, recently 
proposed by the authors. Some of the basic properties of multiset codes are 
established, among which their equivalence to integer codes under the Manhattan 
metric. The presented coding-theoretic framework follows closely the one proposed 
by K\"otter and Kschischang for the operator channels. The two mathematical 
models are similar in many respects, and the basic idea is presented in a way 
which admits a unified view on coding for these types of channels. 
\end{abstract}
%
%
%
%
\section{Introduction}
In this paper, we study the problem of error correction in the \emph{permutation 
channels}. We aim to present a coding-theoretic framework for such channels, which 
is based on the notion of \emph{multiset codes}. These codes are a generalization 
of the so-called \emph{subset codes} recently proposed by the authors \cite{subset}, 
and are argued to be appropriate constructs for reliable information transmission 
over permutation channels. 
\par Permutation channels arise, for example, as models for an end-to-end transmission 
in some types of packet networks. Namely, certain network protocols provide no guarantees 
on the in-order delivery of packets \cite{galag}, and in addition to dropping some packets, 
delivering erroneous packets, etc., have the effect of delivering an essentially random 
permutation of the packets sent. 
Examples include a number of recently popular networking technologies such as mobile 
ad-hoc networks, vehicular networks, delay tolerant networks, wireless sensor networks, etc. 
In the following section we will give a more detailed 
description of the channel model that we consider, as well as the basic idea underlying 
the definition of codes for such a channel. This idea is the same as the one presented 
by K\"otter and Kschischang in their seminal paper \cite{kk}, which gave rise to the 
definition of the operator channel as an appropriate model of random linear network 
coded networks, and codes in projective spaces as adequate constructs for such a channel. 
\par In Section \ref{subset_codes} we will give an overview of the subset coding approach 
presented in \cite{subset}. This approach is then extended and generalized by introducing 
the so-called multiset codes in Section \ref{multiset_codes}. Some basic properties 
of multiset codes and their advantages over subset codes are also described in this 
section. Finally, Section \ref{examples} provides two simple, but fairly general 
examples of both types of codes. 
\section{The channel model}
\label{channel}
Let $S$ be a finite alphabet with $|S|=q>0$ symbols. Without loss of generality, 
we assume that $S=\{1,2,\ldots,q\}$. By a \emph{permutation channel} over $S$ we 
understand the channel whose inputs are sequences of symbols from $S$, and which, 
for any input sequence, outputs a random permutation of this sequence. As noted in 
the Introduction, such channels arise in some types of packet networks in which the 
packets comprising a single message are routed separately and are frequently sent 
over different routes in the network. Therefore, the receiver cannot rely on them 
being delivered in any particular order. 
\par In addition to random permutations, the channel can have other deleterious 
effects on the transmitted sequence, such as \emph{insertions}, \emph{deletions}, 
and \emph{substitutions} of symbols. Substitutions (i.e., errors) are random 
alterations of symbols, usually caused by noise. Insertions and deletions can be 
thought of as synchronization errors, where a symbol is read twice, or is skipped, 
because of the incorrect timing of the receiver's clock. There are also various 
other situations where they occur (see, e.g., \cite{schul}). For example, in a 
networking scenario mentioned above, packet deletions can be caused by network 
congestion and consequent buffer overflows in the routers. Note that, as the 
transmitted sequence is being permuted, \emph{erasures} are essentially the same 
as deletions, because the position of the erased symbol (in the original sequence) 
cannot be deduced. To conclude, the channel considered in this paper is the 
permutation channel with insertions, deletions, and substitutions. 
\begin{remark}
In the case when the permutation channel models a packet network, it should be 
pointed out that the framework proposed here assumes an end-to-end network 
transmission model, and consequently, that coding is done on the transport or application 
layer. It is a frequent assumption in this scenario that only deletions can occur 
in the channel (apart from permutations). Namely, it is understood that errors are 
addressed by error-detecting and error-correcting codes at the lower layers 
(link and physical layer). 
\end{remark}
\subsection{Coding for the permutation channel}
We now state the main idea in a somewhat informal way; the precise definitions 
are given in subsequent sections. 
\par Codes for various types of channel impairments that we consider (insertions, 
deletions, and substitutions) have been thoroughly studied in the literature; but 
how does one deal with random permutations of the symbols? One solution relies on 
the following simple idea: Information should be encoded in an object which is 
invariant under permutations. An example of such an object is a \emph{set}. Based 
on this observation, the authors have introduced the so-called subset codes as 
relevant for the above channel model \cite{subset}, the codewords of which are 
taken to be subsets of an alphabet $S$. An appropriate metric is specified 
on the space of all subsets of $S$, after which the definition of codes and their 
parameters follows familiar lines. In the present paper we further generalize this 
idea by noting that there exists an even more general object invariant under 
permutations -- a \emph{multiset}. Informally, a multiset is a set with repetitions 
of elements allowed. Clearly, for a given alphabet $S$, there are more multisets 
of certain cardinality than there are sets of the same cardinality, and hence, 
this approach can increase the code rate, among other advantages. 
\par To conclude this section, we note that we have adopted the above principle 
of taking codewords to be objects invariant under the channel transformation, from 
the work of K\"otter and Kschischang \cite{kk}. These authors have noticed that this 
principle can be applied to the channels arising in networks which are based on random 
linear network coding (RLNC). In such networks, random linear combinations of the injected 
packets are delivered to the receiver, and hence, the only property preserved by such 
a channel is the vector space spanned by those packets\footnote{\,Actually, it is 
preserved only if the transformation applied to the packets is full-rank, but this 
happens with high probability if the linear combinations are indeed random.}. From 
this observation, the authors of \cite{kk} have developed the notion of \emph{subspace 
codes}, i.e., codes in projective spaces, where codewords are taken to be vector 
subspaces of some ambient vector space (the space of all packets). There are many 
parallels between subspace and subset/multiset codes, as will be evident from the 
exposition in the subsequent sections. In fact, multiset codes can be thought of as 
a generalization of subspace codes in the sense that any basis of a subspace of the 
set of all packets is also a subset of the set of all packets. This is a consequence 
of the fact that the RLNC channel is more restrictive that the permutation channel. 
Namely, permuting the packets is a special case of delivering multiple linear 
combinations of the packets. 
\section{Subset codes}
\label{subset_codes}
Let $S$ be a nonempty finite set representing the alphabet of the given permutation 
channel. If this channel models a packet network, we can think of $S$ as the set of 
all possible packets. Let ${\mathcal P}(S)$ denote the power set of $S$, i.e., the 
set of all subsets of $S$, and ${\mathcal P}(S,\ell)$ the set of all subsets of $S$ 
of cardinality $\ell$. 
\begin{definition}
  A \emph{subset code} over an alphabet $S$ is a non-empty subset of ${\mathcal P}(S)$. 
	If ${\mathcal C}\subseteq{\mathcal P}(S,\ell)$, we say that ${\mathcal C}$ is a 
  constant-cardinality code. 
\end{definition}
\par As usual, in order to enable the receiver to recover from errors, erasures, etc., 
codewords of the code should be chosen to differ from each other as much as possible. 
A measure of "dissimilarity" of sets is needed for this purpose. A natural one, which 
is in fact a metric on ${\mathcal P}(S)$, is given by: 
\begin{equation}
\label{metricI}
  d(X,Y) = |X\bigtriangleup Y|
\end{equation}
for $X,Y\in{\mathcal P}(S)$, where $\bigtriangleup$ denotes the symmetric difference 
between sets which is defined as $X\bigtriangleup Y = (X\setminus Y)\cup(Y\setminus X)$. 
We also have: 
\begin{equation}
 \begin{aligned}
  d(X,Y) &= |X\cup Y| - |X\cap Y|  \\
         &= |X| + |Y| - 2|X\cap Y|  \\
         &= 2|X\cup Y| - |X| - |Y| .
 \end{aligned}
\end{equation}
The distance $d(X,Y)$ is the length of the shortest path between $X$ and $Y$ in the 
Hasse diagram \cite{lattice} of the lattice of subsets of $S$ ordered by inclusion. 
This diagram plays a role similar to the Hamming hypercube for the classical codes 
in the Hamming metric, and it is in fact isomorphic to the Hamming hypercube, as 
discussed below. For constant-cardinality codes, the distance between codewords is 
always even. 
\par The minimum distance of the code can now be defined as: 
\begin{equation}
  \min_{X,Y\in {\mathcal C},\;X\neq Y} d(X,Y). 
\end{equation}
Other important parameters of the code are its size $|{\mathcal C}|$, maximum 
cardinality of the codewords: 
\begin{equation}
  \max_{X\in {\mathcal C}} |X|, 
\end{equation}
and the cardinality of the ambient set, $|S|$. The code ${\mathcal C}\subseteq{\mathcal P}(S)$ 
with minimum distance $d$ and codewords of cardinality at most $\ell$ is said to be 
of type $[\log |S|,\log |{\mathcal C}|,d;\ell]$ (we assume that the logarithms are 
to the base $2$, i.e., that the lengths of the messages are measured in bits). The 
setting we have in mind is the following: The source maps a $k$-bit information sequence 
to a set of $\ell$ $n$-bit symbols/packets which are sent through the channel. In the 
channel, these symbols are permuted, some of them are deleted, some of them are received 
erroneously, and possibly some new symbols are inserted. The receiver collects all these 
symbols and attempts to reconstruct the information sequence. 
\par Having the above scenario in mind, we can also define the rate of an $[n,k,d;\ell]$ 
subset code as: 
\begin{equation}
  R=\frac{k}{n\ell}. 
\end{equation}
\subsection{Isomorphism of subset codes and binary codes}
\label{subset_binary}
When the ambient set $S$ is specified, the subsets of $S$ are uniquely determined by 
their characteristic functions (also called indicator functions). The characteristic 
function of a set $X\subseteq S$ is a mapping $\mathbbm{1}_X : S \to \{0,1\}$, defined 
by: 
\begin{equation}
  \mathbbm{1}_X(x)= \begin{cases}
                      1 & x\in X  \\ 
                      0 & x\notin X . 
                    \end{cases}
\end{equation}
If $S=\{1,\ldots,q\}$, these functions can be identified with binary sequences of 
length $q$, namely $\left(\mathbbm{1}_X(1),\ldots,\mathbbm{1}_X(q)\right)$. All set 
operations (unions, intersections, differences, etc.) on ${\mathcal P}(S)$ can be 
expressed in terms of the corresponding characteristic functions. For example, it 
is easy to see that: 
\begin{equation}
  \mathbbm{1}_{X\bigtriangleup Y} = \mathbbm{1}_X \oplus \mathbbm{1}_Y 
                                  = \left| \mathbbm{1}_X - \mathbbm{1}_Y \right|
\end{equation}
where $\oplus$ denotes the XOR operation (addition modulo 2). Also, the cardinality 
of the set $X$ can be expressed as: 
\begin{equation}
  |X| = \sum\nolimits_{x} \mathbbm{1}_{X}(x), 
\end{equation}
which is just the Hamming weight of the binary sequence 
$\left(\mathbbm{1}_{X}(1),\ldots,\mathbbm{1}_{X}(q)\right)$. 
From the above one concludes that: 
\begin{equation}
  d(X,Y) = |X\bigtriangleup Y| 
         = \sum\nolimits_{x} \left|\mathbbm{1}_X(x) - \mathbbm{1}_Y(x) \right|, 
\end{equation}
i.e., the distance between sets $X$ and $Y$ is equal to the Hamming distance between 
the binary sequences corresponding to $\mathbbm{1}_{X}$ and $\mathbbm{1}_{Y}$. The 
above reasoning implies the following interesting fact: The subset codes in 
${\mathcal P}(S)$ are just a different representation of binary codes in the space 
$\{0,1\}^{|S|}$ under the Hamming metric. Every subset code of type $[n,k,d;\ell]$ 
has a binary counterpart with parameters $\left(2^n,k,d\right)$ and maximum codeword 
weight $\ell$, and vice versa. 
\begin{example}
  Let $S=\{1,2,3,4,5\}$. Any subset of $S$ can be identified by a binary sequence 
  of length $5$; for example $\{1,2\}\leftrightarrow 11000$, $\{2,4\}\leftrightarrow 01010$, 
  etc. Consider now some code in $\{0,1\}^5$, e.g., ${\mathcal C}=\{11000, 01010, 01110, 00111\}$. 
  The subset counterpart of this code is then ${\mathcal C}_\textsc{s}=\left\{\{1,2\},\{2,4\},\{2,3,4\},\{3,4,5\}\right\}$. 
  The distance between two subsets of $S$ is the Hamming distance between the 
  corresponding binary sequences, for example: 
  \begin{equation}
    d\left(\{1,2\},\{2,4\}\right) = \left|\{1,4\}\right| = 2 = d_\textsc{h}(11000,01010)
  \end{equation}
  so that all properties of $\mathcal C$ directly translate into equivalent 
  properties of the subset code ${\mathcal C}_\textsc{s}$. 
\end{example}
\par An important consequence of this isomorphism is that subset codes can be 
constructed by using the familiar constructions of the codes for binary channels. 
Apart from the construction itself, the analogy can be used for the analysis of the 
transmission of a subset through a channel. Namely, an equivalent way of describing 
that $X$ was sent and $Y$ was received, is that the binary word $\left(\mathbbm{1}_X(1),
\ldots,\mathbbm{1}_X(q)\right)$ was sent (through the corresponding binary channel) 
and $\left(\mathbbm{1}_Y(1),\ldots,\mathbbm{1}_Y(q)\right)$ was received. Insertion 
of an element $i\notin X$ to $X$ corresponds to the $0\to1$ transition in the binary 
channel, i.e., $\mathbbm{1}_X(i)=0$ and $\mathbbm{1}_Y(i)=1$. Similarly, deletion of 
an element $i$ from $X$ corresponds to the $1\to0$ transition, and a substitution 
corresponds to both transitions (at different positions) as it is essentially a 
combination of an insertion and a deletion. Consider further the special case when 
only deletions can occur in the channel (recall that this is a frequent model of an 
end-to-end transmission over packet networks). It is easy to conclude from the above 
discussion that this channel is equivalent to the so-called Z-channel in which the 
crossover $1\to0$ occurs with probability $p$ (the probability of deletion), while 
the crossover $0\to1$ never occurs. The analysis of subset codes and the corresponding 
permutation channel with deletions is thus reduced to the analysis of binary codes 
and the binary Z-channel, respectively. Note that, for both these channels, we can 
design a binary code with appropriate parameters. The difference is that, in the 
binary channel we send a codeword (binary sequence) itself, while in the subset 
case, what we send through the channel are the \emph{positions of ones} in this 
codeword. 
\section{Multiset codes}
\label{multiset_codes}
In this section, we generalize subset codes by allowing the codewords to contain 
multiple copies of their elements. This feature is quite natural, because any interesting 
classical code over a finite alphabet contains codewords with multiple occurrences of some 
symbols. In our case, the codewords are sets, and the objects we need -- sets with repetitions 
of elements allowed -- are known as \emph{multisets} \cite{aigner}. A multiset is defined 
with a set of elements it contains, and numbers of occurrences of each element in the set. 
The number of occurrences of an element, called its multiplicity, is assumed to be finite. 
Finally, we note that multisets are also invariant under permutations and hence are suitable 
for the permutation channel. 
\par Let ${\mathcal M}(S)$ denote the collection of all multisets over an alphabet $S$. 
Operations on ${\mathcal M}(S)$, such as union, intersection, difference, etc., are 
straightforward extensions of the corresponding operations on sets. It is easiest to 
illustrate them on a simple example. 
\begin{example}
  Let $X=\{1,2,2,2,3\}$ and $Y=\{1,2,2,3,3,4\}$ be two multisets over $S=\{1,2,3,4\}$. 
  Then $X\cap Y=\{1,2,2,3\}$, $X\cup Y=\{1,2,2,2,3,3,4\}$, $X\setminus Y=\{2\}$, 
  $Y\setminus X=\{3,4\}$. The cardinality of $X$ and $Y$ is $|X|=5$, $|Y|=6$, 
  respectively. 
\end{example}
\par Codes in the space ${\mathcal M}(S)$ are defined analogously to the codes in 
${\mathcal P}(S)$. In the following, ${\mathcal M}(S,\ell)$ denotes the collection 
of all multisets of cardinality $\ell$. 
\begin{definition}
  A \emph{multiset code} over $S$ is a nonempty subset of ${\mathcal M}(S)$. If 
	${\mathcal C}\subseteq{\mathcal M}(S,\ell)$, we say that $\mathcal C$ is a 
	constant-cardinality code. 
\end{definition}
\par Note that ${\mathcal M}(S)$ is an infinite space. It is always assumed, 
however, even if not explicitly stated, that a multiset code is finite. In 
particular, we have in mind multiset codes with an upper bound on the cardinality 
of the codewords, which is a reasonable constraint from the "practical" point 
of view. In any case, we shall mostly deal with constant-cardinality codes 
where this issue does not arise. 
\par It is easy to see that the function $d$ from \eqref{metricI} is a metric on 
${\mathcal M}(S)$, and hence we can define the minimum distance of a multiset code 
in the same way as for subset codes. Other code parameters are also defined in the 
same way as for subset codes and those definitions will not be repeated here. 
\par We now prove a simple, but basic fact about the correcting capabilities of 
multiset codes. The analogous statement for the special case of subset codes is 
proven in \cite{subset}. 
\begin{theorem}
\label{thm}
  A multiset code $\mathcal C$ with minimum distance $d$ is capable of correcting 
  any pattern of $s$ insertions, $\rho$ deletions, and $t$ substitutions, as long 
	as $2(s+\rho+2t)<d$. 
\end{theorem}
\begin{proof}
  Let $X\in {\mathcal C}$ be the multiset which is transmitted through the channel. 
  Let $Y$ be the received multiset. If $\rho$ packets from $X$ have been deleted, and $s$ 
  new packets have been inserted, then we easily deduce that $|X\cap Y|\geq |X|-\rho$ 
  and $|Y|=|X|-\rho+s$. Since each substitution is essentially a combination of 
  one deletion and one insertion, the actual number of deletions and insertions is 
  $\rho+t$ and $s+t$, respectively, wherefrom one concludes that $|X\cap Y|\geq |X|-\rho-t$ 
  and $|Y|=|X|-\rho+s$, and that 
  \begin{equation}
     d(X,Y) = |X| + |Y| - 2|X\cap Y| \leq s + \rho + 2t. 
  \end{equation}
  Now, if the assumption $2(s +\rho + 2t)<d$ holds, then $d(X,Y)\leq\lfloor\frac{d-1}{2}\rfloor$ 
  and therefore $X$ can be recovered from $Y$ by the minimum distance decoder. 
\end{proof}
\par If only deletions can occur in the channel, then $d(X,Y)=\rho$ and the sent 
codeword is recoverable whenever $\rho\leq\lfloor\frac{d-1}{2}\rfloor$. 
\par An obvious advantage that multiset codes have over subset codes is the code 
rate improvement which is a consequence of them being defined in a bigger space: 
\begin{equation}
  \left|{\mathcal M}(S,\ell)\right|=\binom{q+\ell-1}{\ell} > 
                                 \binom{q}{\ell}=\left|{\mathcal P}(S,\ell)\right|. 
\end{equation}
Further, when $|S|=q$ is "small", it is necessary to use multiset codes because, 
unlike subset codes, they allow the cardinality of the codewords to be larger than 
the cardinality of the alphabet. For example, multiset codes with arbitrary minimum 
distance (and hence, arbitrary correction capability) can be defined even over a 
binary alphabet. 
\subsection{Isomorphism of multiset codes and integer codes}
\label{multiset_integer}
The isomorphism between subset codes and binary codes, which has many important 
consequences, as discussed in Section \ref{subset_binary}, also has an appropriate 
generalization in the multiset framework. Namely, multiset codes turn out to be 
equivalent to integer codes under the so-called Manhattan metric, and this equivalence 
is illustrated next. 
\par Multisets over an alphabet $S$ can be described by their \emph{multiplicity 
functions} in the same way subsets are described by their characteristic functions 
(in fact, that is how multisets are usually defined formally \cite{aigner}). The 
multiplicity function of a multiset $X$ over $S$ is a mapping 
$\mathbbm{m}_X : S \to \mathbb{Z}_{\geq0}$, where $\mathbbm{m}_X(x)$ is the number 
of occurrences of $x$ in $X$. Clearly, a multiset is a set if and only if the range 
of its multiplicity function is $\{0,1\}$. Operations on multisets can be expressed 
in terms of their multiplicity functions, for example: 
\begin{equation}
 \begin{aligned}
  \mathbbm{m}_{X\cup Y} &= \max\{\mathbbm{m}_X , \mathbbm{m}_Y\} ,  \\ 
	\mathbbm{m}_{X\cap Y} &= \min\{\mathbbm{m}_X , \mathbbm{m}_Y\} ,  \\ 
	\mathbbm{m}_{X\setminus Y} &= \max\{0 , \mathbbm{m}_X - \mathbbm{m}_Y\} , 
 \end{aligned}
\end{equation}
while the cardinality of a multiset is expressed as: 
\begin{equation}
  |X| = \sum\nolimits_{x} \mathbbm{m}_{X}(x) . 
\end{equation}
If the alphabet is $S=\{1,2,\ldots,q\}$, the multiplicity function of a multiset $X$ 
is uniquely specified by a sequence 
$\left(\mathbbm{m}_X(1),\ldots,\mathbbm{m}_X(q)\right)\in\mathbb{Z}_{\geq0}^q$. 
Therefore, the space ${\mathcal M}(S)$ is essentially equivalent to the space 
$\mathbb{Z}_{\geq0}^q$. Further, the distance between multisets is: 
\begin{equation}
  d(X,Y) = |X\bigtriangleup Y| 
         = \sum\nolimits_{x} \left|\mathbbm{m}_X(x) - \mathbbm{m}_Y(x) \right| , 
\end{equation}
which is the familiar $\ell_1$ distance, also known as the Manhattan metric. Therefore, 
multiset codes are basically just another description of the codes in $\mathbb{Z}_{\geq0}^q$ 
under the Manhattan metric. Constant-cardinality codes are then equivalent to the codes on 
the "sphere" $\left\{(x_1,\ldots,x_q) : x_i\in\mathbb{Z}_{\geq0}, \sum_i x_i = \ell\right\}$. 
\section{Examples of codes for the permutation channel}
\label{examples}
In this section, we describe a simple way to construct subset and multiset codes, 
and discuss some of the properties of the obtained codes. 
\subsection{Example of subset codes}
\par A straightforward way of obtaining codes for the permutation channel is to 
use some classical error-correcting code, and add a sequence number to every symbol 
of the codeword so that the order of symbols can be restored at the receiving side. 
This approach is illustrated below. 
\par Let $\mathcal A$ be a finite alphabet with $|{\mathcal A}|=q$. Observe some 
code ${\mathcal C}$ over $\mathcal A$ with parameters $(\ell,k,d)$, meaning that 
$|{\mathcal C}|=q^k$, the codewords of ${\mathcal C}$ are $q$-ary sequences of length 
$\ell$, and the Hamming distance between any two codewords is at least $d$. For 
any codeword $\boldsymbol{p}=(p_1,\ldots,p_\ell)\in{\mathcal C}$, create a sequence 
$(t_1,\ldots,t_\ell)$, where $t_i=i\circ p_i$ is a new symbol obtained 
by prepending a sequence number to the symbol $p_i$ ($\circ$ denotes the concatenation 
of strings). This mapping is clearly injective and the set of all sequences thus 
obtained defines a code ${\mathcal C}'$ over an alphabet $S=\{1,\ldots,\ell\}\times{\mathcal A}$ 
with parameters $(\ell,k,d)$. The codewords of ${\mathcal C}'$ are invariant under 
permutations, i.e., any permutation of $(t_1,\ldots,t_\ell)$ has the same meaning to 
the receiver because it can recover $(p_1,\ldots,p_\ell)$ from the sequence numbers. 
Therefore, one can imagine the carrier of information being a \emph{set} 
$\left\{t_1,\ldots,t_\ell\right\}$, and hence this simple construction yields an 
example of a subset code ${\mathcal C}_\textsc{s}$ over $S$. The code has $q^k$ 
codewords, each of cardinality $\ell$. The minimum (subset) distance of the code 
is easily determined by observing two codewords: 
\begin{equation}
 \begin{matrix}
  1\circ p_1 & 2\circ p_2 & \ldots & \ell\circ p_\ell \\
  1\circ r_1 & 2\circ r_2 & \ldots & \ell\circ r_\ell .
 \end{matrix}
\end{equation}
It is evident that the cardinality of the intersection of the subset codewords 
$P=\{1\circ p_1,\ldots,\ell\circ p_\ell\}$ and $R=\{1\circ r_1,\ldots,\ell\circ r_\ell\}$ 
is equal to the number of positions where the sequences $\boldsymbol{p}=(p_1,\ldots,p_\ell)$ 
and $\boldsymbol{r}=(r_1,\ldots,r_\ell)$ agree, which is, on the other hand, equal to $\ell$ 
minus the Hamming distance of these two sequences. Therefore, 
\begin{equation}
  d(P,R) = 2d_\textsc{h}(\boldsymbol{p},\boldsymbol{r}), 
\end{equation}
and hence the minimum (subset) distance of ${\mathcal C}_\textsc{s}$ is $2d$. To 
conclude, this construction yields a subset code ${\mathcal C}_\textsc{s}$ of type 
$[\log q\ell,k\log q, 2d; \ell]$. 
\par Note that the decoding procedure for ${\mathcal C}_\textsc{s}$ is the same as 
for $\mathcal C$ once the codeword of $\mathcal C$ is recovered by using sequence 
numbers. Note also that recovering $(p_1,\ldots,p_\ell)$ from 
$\{1\circ p_1,\ldots,\ell\circ p_\ell\}$ reduces deletions to \emph{erasures}, 
while insertions and substitutions are reduced to \emph{errors}. Namely, if 
$i\circ p_i$ has been deleted, the receiver will be able to deduce that the symbol 
at the $i$th position is missing. Similarly, if $j\circ p_j$ has been inserted and 
the receiver now possesses two symbols with the sequence number $j$, it will choose 
one at random, possibly resulting in an error at the $j$th position. Hence, when 
subset codes constructed in this way are used, the permutation channel (over $S$) 
with insertions, deletions, and substitutions, reduces to the classical discrete 
memoryless channel (over $\mathcal A$) with errors and erasures. 
\par The codes described above are, to the best of our knowledge, the only type 
of error-correcting codes for the permutation channel described in the literature 
(see, e.g., the construction of the "outer" code in \cite{schul}). As we have 
illustratred, they are in fact only a special case of the more general notion of 
subset codes. We note that better subset codes can be constructed via the isomorphism 
given in Section \ref{subset_binary}, i.e., by using the familiar constructions of 
binary codes (see also \cite{subset}). 
\subsection{Example of multiset codes}
We next describe a simple construction which yields an example of a multiset code 
(which is not a subset code). It is also is based on "classical" codes and sequence 
numbers. 
\par Again, let $\mathcal A$ be a finite alphabet with $|{\mathcal A}|=q$ symbols, 
and ${\mathcal C}$ a code over $\mathcal A$. For any codeword 
$\boldsymbol{p}=(p_1,\ldots,p_\ell)\in{\mathcal C}$, we create a sequence 
$(t_1,\ldots,t_\ell)$ by prepending sequence numbers to the symbols 
of $\boldsymbol{p}$, but in such a way that runs of identical symbols in $\boldsymbol{p}$ 
are given the same sequence number. For example, the sequence $(a, a, b, b, c, b)$, 
where $a,b,c\in{\mathcal A}$, is mapped to 
$(1\circ a, 1\circ a, 2\circ b, 2\circ b, 3\circ c, 4\circ b)$. 
The obtained sequence is invariant under permutations, and it is easily concluded, 
similarly to the example from the previous subsection, that this procedure yields a 
multiset code ${\mathcal C}_\textsc{m}$ over $S$. The decoding procedure for ${\mathcal C}_\textsc{m}$ 
is again the same as for $\mathcal C$ once the codeword is recovered from the sequence 
numbers. In this case, however, recovering $\boldsymbol{p}$ from $\{i_1\circ p_1,\ldots,i_\ell\circ p_\ell\}$ 
reduces deletions to \emph{deletions}, insertions to either \emph{insertions} or 
\emph{substitutions}, and substitutions to \emph{substitutions} (i.e., \emph{errors}). 
Namely, if the symbol $i_j\circ p_j$ has been deleted, the receiver cannot deduce (in 
general) which symbol has been deleted because there could have been multiple copies 
of this or some other symbols. Similar reasoning applies for the other cases. Therefore, 
the code $\mathcal C$ has to be resilient to insertions, deletions, and substitutions. 
\par Finally, let us determine the parameters of ${\mathcal C}_\textsc{m}$ from those 
of $\mathcal C$. Let $\mathcal C$ be of type $(\ell, k, d)$, where $k$ and $\ell$ are 
as before, and $d$ is the minimum Levenshtein distance \cite{leven}, which is the relevant 
distance measure for insertion/deletion channels (it is defined as the minimum number 
of insertions and deletions that transform one sequence to the other). 
Observe two multiset codewords $P$ and $R$: 
\begin{equation}
 \begin{matrix}
  i_1\circ p_1 & i_2\circ p_2 & \ldots & i_\ell\circ p_\ell \\
  j_1\circ r_1 & j_2\circ r_2 & \ldots & j_\ell\circ r_\ell , 
 \end{matrix}
\end{equation}
where, $(i_m)$ and $(j_m)$ are nondecreasing integer sequences, as explained above. 
Unfortunately, the distance between $P$ and $R$ in general cannot be expressed via 
the Levenshtein (or Hamming) distance between $\boldsymbol{p}$ and $\boldsymbol{r}$, 
and hence the minimum distance of ${\mathcal C}_\textsc{m}$ cannot be inferred from 
$d$. It is easy to conclude, however, that the distance between two multisets obtained 
in this way is greater than or equal to the Levenshtein distance between the original 
sequences, and therefore the code ${\mathcal C}_\textsc{m}$ is of type 
$[\log q\ell,k\log q, d_\textsc{m}; \ell]$, where $d_\textsc{m}\geq d$. 
\par As noted above, one possible decoding procedure for ${\mathcal C}_\textsc{m}$ 
is to first use the sequence numbers to obtain the right ordering of symbols, and 
then apply the decoding algorithm for $\mathcal C$ to the resulting sequence. If 
this procedure is used, then one easily concludes that the number of insertions 
and deletions which can be corrected is at most $\lfloor\frac{d-1}{2}\rfloor$, and 
therefore, the "effective minimum distance" of the code is $d$. 
\par As a final note here, we would like to stress that the above construction merely 
serves as an illustration of a constant-cardinality mutiset code. The general method 
of construction that can be used is via the corresponding constant-weight integer codes 
in the Manhattan metric, as explained in Section \ref{multiset_integer}. It appears, 
however, that these codes have not been studied thoroughly before, and it remains an 
interesting problem for future research to explore further their properties, and obtain 
explicit constructions and decoding algorithms. 
\section{Conclusion}
\label{conclusion}
We have presented a framework for forward error correction in the permutation channels. 
We have introduced multiset codes as relevant constructs for correcting insertions, 
deletions, and substitutions in such channels. Some basic properties of multiset codes 
have been established. The framework presented is analogous to the one introduced recently 
by K\"otter and Kschischang for the operator channels, and can be viewed as its extension. 
As a consequence, a unified view on coding for RLNC networks and multipath routed packet 
networks is obtained. 
\section*{Acknowledgment}
The authors would like to acknowledge the financial support of the Ministry of 
Science and Technological Development of the Republic of Serbia (grants No. 
TR32040 and III44003). 
\end{document}